\documentclass[aip]{revtex4-1}
\usepackage{graphicx}% Include figure files
\usepackage{dcolumn}% Align table columns on decimal point
\usepackage{bm}% bold math
\usepackage{color}
\usepackage{amsmath}
\usepackage{enumerate}
\usepackage{amssymb}
\usepackage{stmaryrd}
\usepackage{epsfig}
\usepackage{hyperref}
\usepackage{comment}
\usepackage{amsthm}
 \newtheorem{lem}{Lemma}
 \newtheorem{prop}{Proposition}

\def\rd{\mathrm{d}}

\begin{document}

%\preprint{}

\title{Conservation laws of the generalized Riemann equations at $N=2,3,4$}
\author{Binfang Gao}
\email{ycugss2008@163.com}
\author{Kai Tian}
\email{tiankai@lsec.cc.ac.cn}
\author{Q. P. Liu}
\email[Author to whom correspondence should be addressed. Electronic mail: ]{qpl@cumtb.edu.cn}
\author{Lujuan Feng}
\affiliation{Department of Mathematics,
China University of Mining and Technology,
Beijing 100083, People's Republic of China}
\date{\today}

\begin{abstract}
In this paper, we present infinitely many conserved densities satisfying particular conservation law $F_{t}=(2uF)_{x}$  for the generalized Riemann equations at $N=2,3,4$. In the $N=2$ case, we also construct conserved densities corresponding to new conservation laws  containing an arbitrary smooth function. In virtue of reductions and/or changes of variables, related conserved densities are obtained for two component Hunter-Saxton equation, Hunter-Saxton equation, Gurevich-Zybin equation and Monge-Ampere equation.
\end{abstract}

\pacs{}
\keywords{generalized Riemann equations, reciprocal transformation, conserved density, reduction}

\maketitle %\maketitle must follow title, authors, abstract and \pacs

\section{Introduction}
The Riemann equation
\begin{equation}\label{riemann}
u_t + uu_x =0 ,
\end{equation}
where $u=u(x,t)$ and subscripts denote specified partial derivatives, also known as inviscid Burgers equation or Hopf equation in literatures, has been extensively studied as a prototype for many phenomena related to hyperbolic systems. \cite{whitham} Due to an interesting observation of Holm and Pavlov, the Riemann equation \eqref{riemann} was recently generalized to \cite{PopPry,GPPP,PAPP}
\begin{equation}\label{hprie}
  \left(\partial_t + u\partial_x\right)^Nu=0, \quad (N\in \mathbb{Z}_{+})
\end{equation}
which could be reformulated as a $N$-component system, namely
\begin{equation}\label{griemann}
u_{j,t} = (1-\delta_{jN})u_{j+1} - u_1 u_{j,x}, \quad (j=1,2,\cdots,N)
\end{equation}
where $u_1 = u$, $u_{j+1}=(\partial_t + u\partial_x)u_j(1\leq j<N)$ and $\delta_{jN}$ is the Kronecker delta. The equation \eqref{hprie}, or equivalently the system \eqref{griemann}, serves as a multi-component generalization of the Riemann equation \eqref{riemann} and will be referred to as the generalized Riemann equation, which has been shown to be closely related to some physically important models. In fact, the generalized Riemann equation at $N=2$ could be viewed as a regular Whitham type system, \cite{GPPP} and reduced to the Gurevich-Zybin system \cite{Pavlov} or an equation describing non-local gas dynamics \cite{brudas} as well. Through a Miura-type transformation, a $N$-component extension of the Hunter-Saxton equation \cite{HS,HZ} was deduced from the system \eqref{griemann}. \cite{Popo}

The generalized Riemann equation \eqref{griemann} is integrable, as certified by a universal matrix Lax representation at arbitrary $N$. \cite{Popo} Furthermore, the bi-Hamiltonian structures for the generalized Riemann equations at $N=2$, $3$ and $4$ were constructed. \cite{PopPry,PAPP} An extraordinary feature of the generalized Riemann equation \eqref{griemann} is attributed to abundant conserved densities of various types, which were constructed by different approaches, for instance, the recursion operator \cite{brudas,Pavlov}, an iterative procedure \cite{PopPry} or the Lax representation. \cite{Popo} It was already pointed out by Popowicz \cite{Popo} that most conserved densities surprisingly satisfy a common conservation law, namely
\begin{equation}\label{gcl}
\partial_{t} (F)=\partial_{x}(-u_1F),
\end{equation}
where $F$ is a certain smooth function of $u_j(1\leq j\leq N)$ and their derivatives with respect to $x$. However, to the best of our knowledge there is still a lack of explanations to this richness of conserved densities.

Some clues from reduced equations of the generalized Riemann equation \eqref{griemann} may lead us to a better understanding about those conserved densities. On the one hand, the Riemann equation \eqref{riemann} was shown by Olver and Nutku to have a class of conservation law, \cite{olvnut} given by
\begin{equation*}
\partial_t\Big(u_x G(z_2,z_3,\cdots,z_n)\Big) = \partial_x\Big(-uu_x G(z_2,z_3,\cdots,z_n)\Big) ,
\end{equation*}
where $z_2 = u_x^{-3}u_{xx}$, $z_{j+1} = u_x^{-1}(\partial_x z_j)(j\geq 2)$, and $G$ is an arbitrary smooth function of its arguments. On the other hand, similar results have recently been achieved for the Hunter-Saxton equation, \cite{LiuTian} which also admits conserved densities involving arbitrary smooth functions. Supported by these facts, conserved densities of such kind are expected for the generalized Riemann equation \eqref{griemann}.

In this paper, we will construct conserved densities for the generalized Riemann equations at $N=2$, $3$ and $4$. With the help of appropriate changes of variables, we will manage to figure out the most general  conserved densities for which equation \eqref{gcl} holds. It turns out that these conserved densities involve arbitrary functions. Similar discussions will be given to few systems related to the generalized Riemann equation \eqref{griemann}, such as the Gurevich-Zybin system, the Monge-Ampere equation and the two-component Hunter-Saxton equation. The paper is organised as follows. In section II and III, we consider the generalized Riemann equation at $N=2$ and the related systems, respectively. In section IV, the generalized Riemann equations at $N=3, 4$ are studied. Last section presents some discussions.

\section{The generalized Riemann equation at $N=2$} \label{sec:2}
We rescale dependent variables as $u_1 = -2u$, $u_2 = 4v$, and write the generalized Riemann equation at $N=2$ as
\begin{subequations}
\begin{align}
u_{t} =& 2uu_{x}-2v,\label{n2na}\\
v_{t} =& 2uv_{x},\label{n2nb}
\end{align}
\end{subequations}
whose general solutions in implicit form have been obtained by quadratures. \cite{Sakovich}

Differentiating both equations \eqref{n2na} and \eqref{n2nb} with respect to $x$ yields
\begin{subequations}
\begin{align}
u_{xt} =& 2uu_{xx} + 2u_x^2 - 2v_x ,\label{n2nax} \\
v_{xt} =& \partial_x(2uv_{x}).\label{n2nbx}
\end{align}
\end{subequations}
Based on the trivial conservation law \eqref{n2nbx}, we introduce a reciprocal transformation
\[
\mathcal{T}_{1}:(x,t,u,v)\mapsto(y,\tau,\bar u,\bar v),
\]
where new variables are defined by
\begin{equation*}
\rd y=v_{x}\rd x+(2uv_{x})\rd t, \quad \rd\tau=\rd t,\quad \bar{u}(y,\tau)=u_x,\quad \bar{v}(y,\tau) = v_x.
\end{equation*}
Applying $\mathcal{T}_1$ to equations \eqref{n2nax} and \eqref{n2nbx} respectively, we obtain a system of ODEs
\begin{subequations}
\begin{align}
\bar{u}_{\tau} =& 2\bar{u}^2-2\bar{v} , \label{n2ta}\\
 \bar{v}_{\tau} =& 2\bar{u}\bar{v}. \label{n2tb}
\end{align}
\end{subequations}
If we change dependent variables as
\begin{equation*}
\mathcal{T}_2:(\bar{u},\bar{v})\mapsto (p,q) = \left(\frac{\bar{u}^2-2\bar{v}}{\bar{v}^2},\frac{\bar{u}}{\bar{v}}\right) ,
\end{equation*}
then the system \eqref{n2ta}-\eqref{n2tb} is linearized to
\begin{subequations}
\begin{align}
p_{\tau}=&0,\label{n2tap}\\
 q_{\tau}=&-2.\label{n2tbq}
\end{align}
\end{subequations}

Now suppose that $F(u_{x},v_{x},\cdots,u_{nx},v_{nx})$ is a conserved density of the system \eqref{n2na}-\eqref{n2nb} such that
\begin{equation}\label{ConLaw}
\partial_{t}\Big(F(u_{x},v_{x},\cdots,u_{nx},v_{nx})\Big)=\partial_{x}\Big(2uF(u_{x},v_{x},\cdots,u_{nx},v_{nx})\Big) ,
\end{equation}
where $u_{nx} = (\partial_x^n u)$ and $v_{nx} = (\partial_x^n v)$. Denote the conserved density $F$ in the coordinates $(y,\tau,p,q)$  as $\widehat{F}$, namely
\begin{equation*}
\widehat{F}(p,q,p_y,q_y,\cdots,p_{(n-1)y},q_{(n-1)y}) = (\mathcal{T}_2\circ\mathcal{T}_1)F(u_{x},v_{x},\cdots,u_{nx},v_{nx}) ,
\end{equation*}
then the equation \eqref{ConLaw} may be rewritten as
\begin{equation}\label{coN2F}
\partial_\tau\Big(\widehat{F}(p,q,p_y,q_y,\cdots,p_{(n-1)y},q_{(n-1)y})\Big)=\frac{4q}{q^2-p}\widehat{F}(p,q,p_y,q_y,\cdots,p_{(n-1)y},q_{(n-1)y})  ,
\end{equation}
which identically holds provided that $ p$ and $q $ solve \eqref{n2tap}-\eqref{n2tbq}.

%To get explicit expressions of conserved densities, we transform it into coordinates $(y,\tau,p,q)$ as and meanwhile convert equation %to

It is easy to see that on the solutions of the system \eqref{n2tap}-\eqref{n2tbq} the equation \eqref{coN2F} reduces  to
\begin{equation*}
0 = \partial_{\tau}(\widehat{F})-\frac{4q}{q^2-p}\widehat{F} = -2\frac{\partial \widehat{F}}{\partial q}-\frac{4q}{q^2-p}\widehat{F} ,
\end{equation*}
which may be solved by the standard method, and yields
\begin{equation*}
\widehat{F}=\frac{1}{q^2-p}G(p,p_{y},q_{y},\cdots,p_{(n-1)y},q_{(n-1)y}) ,
\end{equation*}
where $G$ is an arbitrary smooth function of its arguments. Switching to the coordinates $(x,t,u,v)$ and taking account of the relations
\begin{equation*}
p = u_{x}^2v_{x}^{-2}-2v_{x}^{-1},\quad q = u_{x}v_{x}^{-1},\quad \partial_y = v_x^{-1}\partial_x,
\end{equation*}
%Then sending $\widehat{F}$ back to
we obtain conserved densities with arbitrary function. Indeed,
\begin{equation}\label{F2}
\frac{v_{x}}{2}G\left(\alpha_1,\alpha_2,\beta_2,\cdots,\alpha_n,\beta_n\right)
\end{equation}
solves the conservation law \eqref{ConLaw}, where new variables are defined as
\begin{equation}\label{N2v1}
  \alpha_{1}\equiv u_{x}^2v_{x}^{-2}-2v_{x}^{-1},\quad \beta_{1}\equiv u_{x}v_{x}^{-1},
\end{equation}
and
\begin{equation}\label{N2v2}
  \alpha_{k+1} \equiv (v_{x}^{-1}\partial_{x})\alpha_{k},\quad \beta_{k+1}\equiv (v_{x}^{-1}\partial_{x})\beta_{k}\quad (k\geq1).
\end{equation}

The conservation law \eqref{ConLaw} with the conserved density \eqref{F2} could be directly proved with the following lemmas.

\begin{lem}\label{lemx1}
  When the system\eqref{n2na}-\eqref{n2nb} holds, $ \partial_{t}\alpha_{k}=2u(\partial_{x}\alpha_{k})(k\geq 1),\partial_{t}\beta_{1}=2u(\partial_{x}\beta_{1})-2$ and $\partial_{t}\beta_{k}=2u(\partial_{x}\beta_{k})(k\geq2)$.
\end{lem}
\begin{proof}
The evolution of $\alpha_{1}$ and $\beta_{1}$ is shown by direct calculation.

Both $\alpha_{k}$ and $\beta_{k}$ are defined by the same recursive relation $\gamma_{k+1}=(v _{x}^{-1}\partial_{x})\gamma_{k}(k\geq1)$. If we assume $\partial_{t}\gamma_{k}=2u(\partial_{x}\gamma_{k})+c_{k}$, where $c_{k}$ is a certain constant, then based on the
recursive relation, we have
\begin{equation*}
  \partial_{t}\gamma_{k+1}=-v_{x}^{-2}v_{xt}(\partial_{x}\gamma_{k})+v_{x}^{-1}(\partial_{x}\partial_{t}\gamma_{k})
  =2u(-v_{x}^{-2}v_{2x}(\partial_{x}\gamma_{k})+v_{x}^{-1}(\partial_{x}^2\gamma_{k}))=2u(\partial_{x}\gamma_{k+1}).
\end{equation*}
Hence, the conclusion holds for all ${\alpha_{k}}^{,}s$ and ${\beta_{k}}^{,}s$.
\end{proof}

For any $n\in \mathbb{Z}_{+}$, we introduce a transformation
\[
\Gamma:(u_{x},v_{x},\cdots,u_{nx},v_{nx})\mapsto (\alpha_1,\beta_{1},\cdots,\alpha_{n},\beta_{n}) ,
\]
where $\alpha_k$'s and $\beta_k$'s are defined by equations \eqref{N2v1} and \eqref{N2v2}.
\begin{lem}\label{lemx2}
$\Gamma$ is invertible.
\end{lem}
\begin{proof}
It is sufficient to formulate $(u_{x},v_{x},\cdots,u_{nx},v_{nx})$ in terms of $(\alpha_1,\beta_{1},\cdots,\alpha_{n},\beta_{n})$. From equation \eqref{N2v1}, we immediately obtain
\begin{equation*}
  u_{x}=\frac{2\beta_{1}}{\beta_{1}^2-\alpha_{1}},\quad v_{x}=\frac{2}{\beta_{1}^2-\alpha_{1}}.
\end{equation*}
From the recursive relation \eqref{N2v2}, we have
\begin{equation*}
  \begin{pmatrix}
    \alpha_{k}\\
    \beta_{k}
  \end{pmatrix}
  =\begin{pmatrix}
    2u_{x}v_{x}^{-k-1}          &       2(v_{x}-u_{x}^2)v_{x}^{-k-2}\\
    v_{x}^{-k}                  &      -u_{x}v_{x}^{-k-1}
  \end{pmatrix}
   \begin{pmatrix}
  u_{kx}\\
  v_{kx}
  \end{pmatrix}
  +\mathbb{G}_{1}(u_{x},v_{x},\cdots,u_{(k-1)x},v_{(k-1)x}),
\end{equation*}
where $\mathbb{G}_{1}$ is a 2-dimensional vector function. The $2\times 2$ matrix in the right hand side is non-singular, so we could solve $u_{kx}$ and $v_{kx}$ step by step, and formulate them in terms of $(\alpha_{1},\beta_{1},\cdots,\alpha_{k},\beta_{k})$.
\end{proof}

Lemmas \ref{lemx1} and \ref{lemx2} enable us to construct conserved densities explicitly depending on $u$ and $v$ for the system \eqref{n2na}-\eqref{n2nb}.  Suppose that $F(u,v,u_{x},v_{x},\cdots,u_{nx},v_{nx})$ is a conserved density such that
\begin{equation}\label{ConLaw2}
  \partial_{t}\left(F(u,v,u_{x},v_{x},\cdots,u_{nx},v_{nx})\Big)=\partial_{x}\Big(2uF(u,v,u_{x},v_{x},\cdots,u_{nx},v_{nx})\right).
\end{equation}
According to lemma \ref{lemx2}, any $F(u,v,u_{x},v_{x},\cdots,u_{nx},v_{nx})$ could be reformulated in terms of $(u,v,\alpha_{1},\beta_{1},\cdots,\alpha_{n},\beta_{n})$ as
\begin{equation*}
\hat F(u,v,\alpha_{1},\beta_{1},\cdots,\alpha_{n},\beta_{n})=\Gamma F(u,v,u_{x},v_{x},\cdots,u_{nx},v_{nx}),
\end{equation*}
and with the help of lemma \ref{lemx1}, the conservation law \eqref{ConLaw2} could be calculated as
\begin{equation*}
  0=\partial_{t}\hat F-\partial_{x}(2u\hat F)=-2v\frac{\partial\hat F}{\partial u}-2\frac{\partial\hat F}{\partial \beta_{1}}-\frac{4\beta_{1}}{\beta_{1}^2-\alpha_{1}}\hat F.
\end{equation*}
Solving it by the standard method gives us conserved densities
\begin{equation*}
\hat F= \frac{v_{x}}{2}G(v,\alpha_{1},\beta_{1}-uv^{-1},\alpha_{2},\beta_{2},\cdots,\alpha_{n},\beta_{n}),
\end{equation*}
which takes conserved densities of the form \eqref{F2} as special cases. The result is summarized as
\begin{prop}\label{prop:1}
On solutions of the generalized Riemann equation at $N=2$ \eqref{n2na}-\eqref{n2nb}, the conservation law \eqref{ConLaw2} holds if and only if
\begin{equation*}
F(u,v,u_{x},v_{x},\cdots,u_{nx},v_{nx}) = \frac{v_{x}}{2}G(v,\alpha_{1},\beta_{1}-uv^{-1},\alpha_{2},\beta_{2},\cdots,\alpha_{n},\beta_{n}),
\end{equation*}
where $\alpha_k$'s and $\beta_k$'s are defined by equations \eqref{N2v1} and \eqref{N2v2}, and $G$ is an arbitrary smooth function of its arguments.
\end{prop}

In the rest part of the section, we manage to find new conservation laws other than \eqref{ConLaw2} for the system \eqref{n2na}-\eqref{n2nb}. Let us consider a smooth function
\begin{equation*}
  H(u_{x},v_{x},\cdots,u_{nx},v_{nx})=v_{x}P(\alpha_{1},\beta_{1},\cdots,\alpha_{n},\beta_{n}).
\end{equation*}
With the aid of lemma \ref{lemx1}, it is straightforward to have
\begin{align*}
  \partial_{t}H =\partial_{x}(2uv_{x}P)-2v_{x}\frac{\partial P}{\partial \beta_{1}}.
\end{align*}
If there exists $Q(\alpha_{1},\beta_{1},\cdots,\alpha_{n-1},\beta_{n-1})$ such that
\begin{equation}\label{eqQ}
\partial_{x}Q = -2v_{x}\frac{\partial P}{\partial \beta_{1}},
\end{equation}
then $ H=v_{x}P(\alpha_{1},\beta_{1},\cdots,\alpha_{n},\beta_{n})$ would be a conserved density such that
\begin{equation*}
  \partial_{t}H=\partial_{x}(2uH+Q).
\end{equation*}
Taking account of the recursive relation \eqref{N2v2}, we have
\begin{equation*}
\partial_{x}Q = \sum_{k=1}^{n-1}\left(\frac{\partial Q}{\partial\alpha_{k}}(\partial_x\alpha_k)+\frac{\partial Q}{\partial\beta_{k}}(\partial_x\beta_k)\right) = v_x\sum_{k=1}^{n-1}\left(\frac{\partial Q}{\partial\alpha_{k}}\alpha_{k+1}+\frac{\partial Q}{\partial\beta_{k}}\beta_{k+1}\right).
\end{equation*}
Then equation \eqref{eqQ} is rewritten as
\begin{equation*}
\frac{\partial P}{\partial \beta_{1}} = -\frac{1}{2}\sum_{k=1}^{n-1}\left(\frac{\partial Q}{\partial\alpha_{k}}\alpha_{k+1}+\frac{\partial Q}{\partial\beta_{k}}\beta_{k+1}\right).
\end{equation*}

We summarize above discussions as
\begin{prop}\label{prop:2}
When the generalized Riemann equation at $N=2$ \eqref{n2na}-\eqref{n2nb} holds, given an arbitrary smooth function $Q(\alpha_{1},\beta_{1},\cdots,\alpha_{n-1},\beta_{n-1})$, let
\begin{equation}\label{eqP}
P = -\frac{1}{2}\int \sum_{k=1}^{n-1}\left(\frac{\partial Q}{\partial\alpha_{k}}\alpha_{k+1}+\frac{\partial Q}{\partial\beta_{k}}\beta_{k+1}\right)\rd\beta_1,
\end{equation}
then $H = v_xP$ is conserved and satisfies the conservation law $\partial_{t}H=\partial_{x}(2uH+Q)$.
\end{prop}

As an implementation of proposition \ref{prop:2}, we construct a conserved density from the smooth function $Q(\alpha_1)\beta_1$. Following the formula \eqref{eqP}, we have
\begin{equation*}
P = -\frac{1}{2}\int \Big(Q^\prime(\alpha_1)\alpha_2\beta_1 + Q(\alpha_1)\beta_2\Big)\rd\beta_1 = -\frac{1}{4}Q^\prime(\alpha_1)\alpha_2\beta_1^2 -\frac{1}{2}Q(\alpha_1)\beta_1\beta_2 + T(\alpha_1,\alpha_2,\beta_2) ,
\end{equation*}
where $Q^\prime(\alpha_1)=\partial Q/\partial\alpha_1$ and $T(\alpha_1,\alpha_2,\beta_2)$ is the ``constant''  of integration. Then, we obtain a conserved density
\begin{equation*}
H(u_x,v_x,u_{xx},v_{xx}) = -\frac{1}{4}v_xQ^\prime(\alpha_1)\alpha_2\beta_1^2 -\frac{1}{2}v_xQ(\alpha_1)\beta_1\beta_2 + v_xT(\alpha_1,\alpha_2,\beta_2)
\end{equation*}
such that
\begin{equation}\label{ConLaw3}
\partial_tH(u_x,v_x,u_{xx},v_{xx}) = \partial_x(2uH(u_x,v_x,u_{xx},v_{xx})+Q(\alpha_1)\beta_1).
\end{equation}
We notice that $v_xT(\alpha_1,\alpha_2,\beta_2)$ is also a conserved density as implied by proposition \ref{prop:1}, and satisfies
\begin{equation*}
\partial_t(v_xT(\alpha_1,\alpha_2,\beta_2)) = \partial_x(2uv_xT(\alpha_1,\alpha_2,\beta_2)).
\end{equation*}
So the conservation law \eqref{ConLaw3} could be simplified by replacing $H(u_x,v_x,u_{xx},v_{xx})$ by
\begin{equation*}
-\frac{1}{4}v_xQ^\prime(\alpha_1)\alpha_2\beta_1^2 -\frac{1}{2}v_xQ(\alpha_1)\beta_1\beta_2.
\end{equation*}

To conclude this section, we comment on that regarding the generalized Riemann equation at $N=2$ \eqref{n2na}-\eqref{n2nb}, while above two propositions provide most conserved densities, they by no means exhaust all possibilities since there are extra conservation laws explicitly depending on time, \cite{Popo} such as
  \begin{align*}
    &\partial_{t}(2tu_{x}v+v)=\partial_{x}\big(t(4uu_{x}v-2v^2)+2uv\big),\\
    &\partial_{t}(u+2tv+2t^2u_{x}v)=\partial_{x}\big(t^2(4uu_{x}v-2v^2)+4tuv+u^2\big).
  \end{align*}

\section{Conservation laws of related systems}
Under appropriate changes of variables and/or reductions, the generalized Riemann equation at $N=2$ \eqref{n2na}-\eqref{n2nb} is converted/reduced to the Gurevich-Zybin system, Monge-Ampere equation,  two-component Hunter-Saxton equation and Hunter-Saxton equation. In this section, we will show that conserved densities of these systems could be easily deduced from those presented in section \ref{sec:2}.

\subsection{The Gurevich-Zybin system and Monge-Ampere equation}
For the generalized Riemann equation at $N=2$ \eqref{n2na}-\eqref{n2nb}, we introduce $u = -\hat{u}/2$ and $v = -\Phi_x/4$. Then equation \eqref{n2na} is changed to
\begin{equation}\label{gzu}
\hat{u}_{t}+\hat{u}\hat{u}_{x}+\Phi_{x}=0,
\end{equation}
while equation \eqref{n2nb} to
\begin{equation}\label{gztp}
\Phi_{xt} + \hat{u}\Phi_{xx} =0.
\end{equation}
Differentiating equation \eqref{gztp} with respect to $x$ and letting $\rho=\Phi_{xx}$, we obtain
\begin{equation}\label{gzrho}
\rho_{t}+\partial_{x}(\hat{u}\rho)=0 .
\end{equation}
The coupled system \eqref{gzu} and \eqref{gzrho} (N.B. $\rho=\Phi_{xx}$) is nothing but the Gurevich-Zybin system in 1-dimensional space,\cite{GZ1,GZ2} which has been shown to be linearisable and possesses infinitely many local Hamiltonian structures, local Lagrangian representations and local conservation laws (symmetries).\cite{Pavlov} As observed by Pavlov, solving $\hat{u}$ from equation \eqref{gztp} and substituting it into a conservation law of the Gurevich-Zybin system, i.e.
\begin{equation*}
(\rho \hat u)_{t}+(\rho {\hat u}^2+\Phi_{x}^2/2)_{x}=0 ,
\end{equation*}
we obtain the Monge-Ampere equation
\begin{equation}\label{monge}
  \Phi_{tt}=\frac{\Phi_{xt}^2}{\Phi_{xx}}+\frac{1}{2}\Phi_{x}^2.
\end{equation}

Taken the above connection into consideration, their conserved densities are easily deduced from those of the system \eqref{n2na}-\eqref{n2nb}. For instance, inferred from the conserved density \eqref{F2}, the Gurevich-Zybin system \eqref{gzu} and \eqref{gzrho} has conserved densities of the form
\begin{equation*}
\rho G(\bar{\alpha}_{1},\bar{\alpha}_{2},\bar{\beta}_{2},\cdots,\bar{\alpha}_{n},\bar{\beta}_{n}),
\end{equation*}
such that the conservation law
\begin{equation*}
\partial_{t}(\rho G)=\partial_{x}(-\hat u\rho G),
\end{equation*}
where $\bar{\alpha}_{1}=(\hat u_{x}^2+2\rho)\rho^{-2}$, $\bar{\beta}_{1}=\hat u_{x}\rho^{-1}$ and
\begin{equation*}
\bar{\alpha}_{k+1}=(\rho^{-1}\partial_{x})\bar{\alpha}_{k}, \quad \bar{\beta}_{k+1}=(\rho^{-1}\partial_{x})\bar{\beta}_{k}\quad(k\ge1).
\end{equation*}
Similarly, replacing $u$ by $\Phi_{xt}/(2\Phi_{xx})$, while $v$ by $-\Phi_x/4$ in the conserved density \eqref{F2} yields conserved densities of the Monge-Ampere equation \eqref{monge}.

\subsection{Two-component Hunter-Saxton equation}
Let $v_x  = (u_x^2 + \eta^2)/2$, then the system \eqref{n2nax}-\eqref{n2nbx} is converted to
\begin{subequations}
\begin{align}
u_{xt}&=2uu_{xx}+u_{x}^2-\eta^2,\label{tHS1} \\
\eta_{t}&=2(u\eta)_{x},\label{tHS2}
\end{align}
\end{subequations}
which obviously reduces to the Hunter-Saxton equation \cite{HS,HZ} when $\eta=0$, and is referred to as two-component Hunter-Saxton equation. \cite{Pavlov,PopPry}

Replacing $v_x$ by $(u_x^2 + \eta^2)/2$ in \eqref{F2} immediately gives us conserved densities
\begin{equation}\label{HScon}
(u_x^2 + \eta^2)G_{1}(\tilde{\alpha}_1,\tilde{\alpha}_2,\tilde{\beta}_2,\cdots,\tilde{\alpha}_n,\tilde{\beta}_n)
\end{equation}
for the two-component Hunter-Saxton equation \eqref{tHS1}-\eqref{tHS2},
where $\tilde{\alpha}_1 = \eta^2(u_{x}^2+\eta^2)^{-2}$, $\tilde{\beta}_1=u_{x}(u_{x}^2+\eta^2)^{-1}$ and
\begin{equation*}
\tilde{\alpha}_{k+1} = (u_{x}^2+\eta^2)^{-1}\partial_x\tilde{\alpha}_{k},\quad
\tilde{\beta}_{k+1} = (u_{x}^2+\eta^2)^{-1}\partial_x\tilde{\beta}_{k}\quad (k\geq 1)
\end{equation*}
and the corresponding conservation law is given by
\begin{equation*}
\partial_{t}\left((u_{x}^2+\eta^2)G_{1}\Big)=\partial_{x}\Big(2u(u_{x}^2+\eta^2)G_{1}\right)
\end{equation*}
When $\eta=0$, all $\tilde{\alpha}_{k}$'s vanish and the quantity \eqref{HScon} reduces to
\begin{equation*}
u_{x}^2  G_{2}\Big((u_x^{-2}\partial_x)u_x^{-1},(u_x^{-2}\partial_x)^2u_x^{-1},\cdots,(u_x^{-2}\partial_x)^{n-1}u_x^{-1}\Big),
\end{equation*}
where $G_2=G_1|_{\eta=0}$. Thus we recover the conserved density with arbitrary function of the Hunter-Saxton equation recently reported by Tian and Liu. \cite{LiuTian}

\section {The generalized Riemann equations at $N=3,4$}
In this section, we will construct conserved densities with arbitrary smooth functions for the generalized Riemann equations at $N=3,4$. To this end, we will introduce new dynamical variables in such a way that their evolutions take  simple forms. All results could be verified as we did in section \ref{sec:2}, so detailed proofs are omitted.

\subsection{$N=3$}
By rescaling dependent variables as $u_{1}=-2u,u_{2}=4v$ and $u_{3}=-8w$, the generalized  Riemann system at $N=3$ is rewritten as
\begin{subequations}
\begin{align}
u_{t}&=2uu_{x}-2v,\label{N3a}\\
v_t&=2uv_{x}-2w,\label{N3b}\\
w_t&=2uw_{x}.\label{N3c}
\end{align}
\end{subequations}
To get conserved densities with concise expressions, let
\begin{equation*}\label{N3v1}
\alpha_{1}\equiv(v_{x}^2-2u_{x}w_{x})w_{x}^{-2},\quad\beta_{1}\equiv(v_{x}^3+3w_{x}^2-3u_{x}v_{x}w_{x})w_{x}^{-3},
\quad \gamma_{1}\equiv v_{x}w_{x}^{-1},
\end{equation*}
and
\begin{equation*}\label{N3v2}
  \alpha_{k+1}\equiv(w_{x}^{-1}\partial_{x})^{k}\alpha_{1},\quad \beta_{k+1}\equiv(w_{x}^{-1}\partial_{x})^{k}\beta_{1},\quad \gamma_{k+1}\equiv(w_{x}^{-1}\partial_{x})^{k}\gamma_{1}\quad (k\geq1)
\end{equation*}
then in the same way as we proved lemmas \ref{lemx1} and \ref{lemx2}, we obtain
\begin{lem}\label{lemx3}
 when the system \eqref{N3a}-\eqref{N3c} holds, $\partial_{t}\alpha_{k}=2u(\partial_{x}\alpha_{k})$, $\partial_{t}\beta_{k}=2u(\partial_{x}\beta_{k})(k\geq1)$, $\partial_{t}\gamma_{1}=2u(\partial_{x}\gamma_{1})-2$ and
 $\partial_{t}\gamma_{k}=2u(\partial_{x}\gamma_{k})(k\geq2)$.
\end{lem}
\begin{lem}\label{lemx4}
For any $n\in\mathbb{Z}_{+}$, the transformation $$\Gamma_{2}:(u_{x},v_{x},w_{x},\cdots,u_{nx},v_{nx},w_{nx})\mapsto (\alpha_{1},\beta_{1},\gamma_{1},\cdots,\alpha_{n},\beta_{n},\gamma_{n})$$ is invertible.
\end{lem}

\begin{comment}
\begin{proof}
 When n=1, it is easy to obtain from the relation \eqref{N3v1}that
 \begin{equation*}
   u_{x}=\frac{3(\gamma_{1}^2-\alpha_{1})}{\gamma_{1}^3-3\alpha_{1}\gamma_{1}+2\beta_{1}},\quad
   v_{x}=\frac{6\gamma_{1}}{\gamma_{1}^3-3\alpha_{1}\gamma_{1}+2\beta_{1}},\quad
   w_{x}=\frac{6}{\gamma_{1}^3-3\alpha_{1}\gamma_{1}+2\beta_{1}},
 \end{equation*}
 Thanks to the equation \eqref{N3v2}, we have
 \begin{equation*}
   \begin{pmatrix}
   \alpha_{k}& \beta_{k} &\gamma_{k}
   \end{pmatrix}^{\mathrm{T}}
   =M_{1}\begin{pmatrix}
   u_{kx} & v_{kx} & w_{kx}
   \end{pmatrix}^{\mathrm{T}}
   +\mathbb{G}_{2}(u_{x},v_{x},w_{x},\cdots,u_{(k-1)x},v_{(k-1)x},w_{(k-1)x}),
 \end{equation*}
 where $(\bullet)^{\mathrm{T}}$ denotes the transpose of the Matrix, $\mathbb{G}_{2}$ is a 3-dimensional vector function and
 \begin{equation*}
 M_{1}=
 \begin{pmatrix}
   -2w_{x}^{-k} & 2v_{x}w_{x}^{-k-1} & 2(u_{x}w_{x}-v_{x}^2)w_{x}^{-k-2}\\
   -3v_{x}w_{x}^{-k-1} & 3(v_{x}^2-u_{x}w_{x})w_{x}^{-k-2} & 3(2u_{x}v_{x}w_{x}-v_{x}^3-w_{x}^2)w_{x}^{-k-3}\\
   0 & w_{x}^{-k} &  -v_{x}w_{x}^{-k-1}
 \end{pmatrix}.
 \end{equation*}
  Owing to the non-singularity of $M_{1}$, we could arrive at the $u_{kx},v_{kx}$ and $w_{kx}$ by solving the algebraic equations and represent them in terms of $(\alpha_{1},\beta_{1},\gamma_{1},\cdots,\alpha_{k},\beta_{k},\gamma_{k})$.
\end{proof}
\end{comment}

Regarding the generalized Riemann equation at $N=3$ \eqref{N3a}-\eqref{N3c}, suppose that $$F = F(u,v,w,u_{x},v_{x},w_{x},\cdots,u_{nx},v_{nx},w_{nx})$$ is a conserved density such that
\begin{equation}\label{n3con}
\partial_{t}F=\partial_{x}(2uF).
\end{equation}
According to lemmas \ref{lemx3} and \ref{lemx4}, the conservation law \eqref{n3con} could be solved in coordinates $(u,v,w,\alpha_{1},\beta_{1},\gamma_{1},\cdots,\alpha_{k},\beta_{k},\gamma_{k})$.

\begin{prop}
Equation \eqref{n3con} is qualified as a conservation law of the system \eqref{N3a}-\eqref{N3c} if and only if
  \begin{equation*}
    F=w_{x}G\Big(w,v^2-2wu,\alpha_{1},\beta_{1},\gamma_{1} - w^{-1}v,
    \alpha_{2},\beta_{2},\gamma_{2},\cdots,\alpha_{n},\beta_{n},\gamma_{n}\Big).
  \end{equation*}
\end{prop}

In the spirit of proposition \ref{prop:2}, we also find conservation laws of the form
\begin{equation*}
  \partial_{t}\left(w_{x}P\Big)=\partial_{x}\Big(2uw_{x}P+Q\right)_{x},
\end{equation*}
where $Q=Q(\alpha_{1},\beta_{1},\gamma_{1},\cdots,\alpha_{n-1},\beta_{n-1},\gamma_{n-1})$ is an arbitrary smooth function and
\begin{equation*}
 P =-\frac{1}{2}\int\sum_{k=1}^{n-1}\left(\frac{\partial Q}{\partial \alpha_{k}}\alpha_{k+1}+\frac{\partial Q}{\partial \beta_{k}}\beta_{k+1}+\frac{\partial Q}{\partial \gamma_{k}}\gamma_{k+1}\right)d\gamma_{1}.
\end{equation*}

\subsection{$N=4$}
The generalized Riemann equation at $N=4$ is given by
\begin{subequations}
\begin{align}
u_{t}&=2uu_{x}-2v,\label{N4a}\\
v_t&=2uv_{x}-2w,\label{N4b}\\
w_t&=2uw_{x}-2z,\label{N4c}\\
z_t&=2uz_{x},\label{N4d}
\end{align}
\end{subequations}
where $u=-u_1/2$, $v=u_2/4$, $w=-u_3/8$ and $z = u_4/16$. Let
\begin{eqnarray*}
&\alpha_{1}\equiv(w_{x}^2-2v_{x}z_{x})z_{x}^{-2},\quad \beta_{1}\equiv(v_{x}^2-2u_{x}w_{x}+2z_{x})z_{x}^{-2},\label{N4v1}\\ &\gamma_{1}\equiv(w_{x}^3-3z_{x}v_{x}w_{x}+3z_{x}^2u_{x})z_{x}^{-3},
\quad \delta_{1}\equiv w_{x}z_{x}^{-1},\label{N4v2}
\end{eqnarray*}
and
\begin{eqnarray*}
& \alpha_{k+1}\equiv (z_{x}^{-1}\partial_{x})^{k}\alpha_{1},\qquad \beta_{k+1}\equiv(z_{x}^{-1}\partial_{x})^{k}\beta_{1},\label{N4v3} \\ \gamma_{k+1}
& \equiv(z_{x}^{-1}\partial_{x})^{k}\gamma_{1},\qquad \delta_{k+1}\equiv(z_{x}^{-1}\partial_{x})^{k}\delta_{1}\quad (k\geq 1)\label{N4v4}.
\end{eqnarray*}
then in terms of these new dynamical variables, we have
\begin{lem}\label{lemx5}
When the system \eqref{N4a}-\eqref{N4d} holds, $\partial_{t}\alpha_{k}=2u(\partial_{x}\alpha_{k})$, $\partial_{t}\beta_{k}=2u(\partial_{x}\beta_{k})$, $\partial_{t}\gamma_{k}=2u(\partial_{x}\gamma_{k})(k\geq1)$,
$\partial_{t}\delta_{1}=2u(\partial_{x}\delta_{1})-2$ and $\partial_{t}\delta_{k}=2u(\partial_{x}\delta_{k})(k\geq2)$.
\end{lem}
\begin{lem}\label{lemx6}
For any $n\in\mathbb{Z}_{+}$, the transformation
\begin{equation*}
\Gamma_{3}:(u_{x},v_{x},w_{x},z_{x},\cdots,u_{nx},v_{nx},w_{nx},z_{nx})\mapsto (\alpha_{1},\beta_{1},\gamma_{1},\delta_{1},\cdots,\alpha_{n},\beta_{n},\gamma_{n},\delta_{n})
\end{equation*}
is invertible.
\end{lem}

On the basis of above lemmas, the following results are proved.
\begin{prop}
$F = F(u,v,w,z,\cdots,u_{nx},v_{nx},w_{nx},z_{nx})$ is a conserved density such that $\partial_{t}F=\partial_{x}(2uF)$ if and only if
\begin{equation*}
  F=z_{x}G\left(z,w^2-2vz,\frac{3uz^2-3vwz+w^3}{z^2},\alpha_{1},\beta_{1},\gamma_{1},\delta_{1}-\frac{w}{z},
\alpha_{2},\beta_{2},\gamma_{2},\delta_{2},\cdots,\alpha_{n},\beta_{n},\gamma_{n},\delta_{n}\right).
\end{equation*}
\end{prop}
\begin{prop}
For any smooth function $Q = Q(\alpha_{1},\beta_{1},\gamma_{1},\delta_{1},\cdots,\alpha_{n-1},\beta_{n-1},\gamma_{n-1},\delta_{n-1})$,
\begin{equation*}
  \partial_{t}\left(z_{x}P\Big)=\partial_{x}\Big(2uz_{x}P+Q\right),
\end{equation*}
is a conservation law of the system \eqref{N4a}-\eqref{N4d} if
\begin{equation*}
 P = -\frac{1}{2}\int\sum_{k=1}^{n-1}\left(\frac{\partial Q}{\partial \alpha_{k}}\alpha_{k+1}+\frac{\partial Q}{\partial \beta_{k}}\beta_{k+1}+\frac{\partial Q}{\partial \gamma_{k}}\gamma_{k+1}+\frac{\partial Q}{\partial \delta_{k}}\delta_{k+1}\right)d\delta_{1}.
 \end{equation*}
\end{prop}

\section{Conclusion}
It is well known that the existence of infinitely many conserved densities dependent on arbitrary function  implies strong integrability for a system. In this paper, by introducing special variables we obtain infinitely many conserved densities (two categories) for the generalized Riemann equations at $N=2,3$ and 4. Through the changes of variables or reductions, conserved densities for related systems are obtained. It is not difficult to find that the introduced variables such as $\alpha_{k},\beta_{k},\gamma_{k},\delta_{k}$ whose evolutions are simple all  take rational form and  could be endowed with degree. Thus it may be possible to consider the generalized Riemann equation for the general $N$.  However in such case we do not have ready results.
%There remains a problem to further investigations of  conserved densities for the generalized Riemann equation for arbitrary $N$.
\section*{acknowledgments}
This paper is supported by the National Natural Science Foundation of China (grant numbers: 11271366, 11331008 and 11505284) and the Foundamental Research Funds for Central Universities.

%\bibliographystyle{plain}% amsalpha
%\bibliography{dsw}

\begin{thebibliography}{99}
\bibitem{whitham} Whitham, G. B., ``Linear and nonlinear waves ''(Wiley, Canada, 1999).

\bibitem{PopPry} Popowicz, Z. and Prykarpatsky, A. K., ``The non-polynomial conservation laws and integrability analysis of generalized Riemann type hydrodynamical equations,'' Nonlinearity. {\bf 23}(10), 2517--2537 (2010).

\bibitem{GPPP} Golenia, J., Pavlov, M. V., Popowicz, Z., and Prykarpatsky, A. K., ``On a nonlocal Ostrovsky-Whitham type dynamical system, its Riemann type inhomogeneous regularizations and their integrability,'' SIGMA. {\bf 6}, 002 (2010).

\bibitem{PAPP} Prykarpatsky, A. K., Artemovych, O. D., Popowicz, Z., and Pavlov, M. V., ``Differential-algebraic integrability analysis of the generalized Riemann type and Korteweg-de Vries hydrodynamical equations,'' J. Phys. A: Math. Theor. {\bf 43}(29), (2010)295205.

\bibitem{Pavlov}Pavlov, M. V., ``The Gurevich-Zybin system,'' J. Phys. A: Math. Gen. {\bf 38}(17), 3823--3840 (2005).

\bibitem{brudas} Brunelli, J. C. and Das, A., ``On an integrable hierarchy derived from the isentropic gas dynamics,'' J. Math. Phys. {\bf  45}(7), 2633--2645 (2004).

\bibitem{HS}Hunter, J. K. and Saxton, R., ``Dynamics of director fields,'' SIAM J. Appl. Math. {\bf 51}(6), 1498--1521 (1991).

\bibitem{HZ}Hunter, J. K. and Zheng, Y., ``On a completely integrable nonlinear hyperbolic variational equation,'' Physica. D. {\bf 79}(2-4), 361--384 (1994).

\bibitem{Popo} Popowicz, Z., ``The matrix Lax representation of the generalized Riemann equations and its conservation laws,'' Phys. Lett. A. {\bf 375}(37), 3268--3272 (2011).

\bibitem{olvnut} Olver, P. J. and Nutku, Y., ``Hamiltonian structures for systems of hyperbolic conservation laws,'' J. Math. Phys. {\bf 29}(7), 1610--1619 (1988).

\bibitem{LiuTian} Tian, K. and Liu, Q. P., ``Conservation laws and symmetries of Hunter-Saxton equation: revisited,'' Nonlinearity. {\bf 29}(3), 737--755 (2016).

\bibitem{Sakovich} Sakovich, S., ``On a Whitham-type equation,'' SIGMA. {\bf 5}, 101 (2009).

\bibitem{GZ1} Gurevich, A. V. and Zybin, K. P., ``Nondissipative gravitational turbulence,'' Sov. Phys. JETP. {\bf 67}(1), 1--12 (1988).

\bibitem{GZ2} Gurevich, A. V. and Zybin, K. P., ``Large-scale structure of the Universe. Analytic theory,'' Sov. Phys. Usp. {\bf 38}(7), 687--722 (1995).
\end{thebibliography}

\end{document}